\documentclass[10pt]{amsart}
\usepackage{amsfonts,amssymb}
\RequirePackage{ifpdf}
\usepackage{amsmath}
\usepackage{color}
\usepackage{pstcol,pst-node}
\usepackage{graphics}
\usepackage{epic}
\usepackage{curves}

\def\appendix#1{
\addtocounter{section}{1} \setcounter{equation}{0}
\renewcommand{\thesection}{\Alph{section}}
\section*{Appendix \thesection\protect\indent\quad
#1}
}
\renewcommand{\theequation}{\thesection.\arabic{equation}}

\catcode`\@=11
\def\marginnote#1{}

\newcount\hour
\newcount\minute
\newtoks\amorpm
\hour=\time\divide\hour by60 \minute=\time{\multiply\hour by60
\global\advance\minute by-\hour}
\edef\standardtime{{\ifnum\hour<12 \global\amorpm={am}%
        \else\global\amorpm={pm}\advance\hour by-12 \fi
        \ifnum\hour=0 \hour=12 \fi
        \number\hour:\ifnum\minute<10 0\fi\number\minute\the\amorpm}}
\edef\militarytime{\number\hour:\ifnum\minute<100\fi\number\minute}

\def\draftlabel#1{{\@bsphack\if@filesw {\let\thepage\relax
      \xdef\@gtempa{\write\@auxout{\string
          \newlabel{#1}{{\@currentlabel}{\thepage}}}}}\@gtempa \if@nobreak
    \ifvmode\nobreak\fi\fi\fi\@esphack} \gdef\@eqnlabel{#1}}
    \def\@eqnlabel{}
\def\@vacuum{}
\def\draftmarginnote#1{\marginpar{\raggedright\scriptsize\tt#1}}

\def\draft{
  \oddsidemargin -.5truein
  \def\@oddfoot{\footnotesize \sl preliminary draft \hfil
    \rm\thepage\hfil\sl\today\quad\militarytime}
  \let\@evenfoot\@oddfoot \overfullrule 3pt
    \let\label=\draftlabel
    \let\marginnote=\draftmarginnote
  \def\@eqnnum{(\theequation)\rlap{\kern\marginparsep\tt\@eqnlabel}%
    \global\let\@eqnlabel\@vacuum}

  }

\newcommand{\tr}{\,{\rm Tr}\,}
\newcommand{\ID}{1\!\!1}

\def\be{\begin{equation}}
\def\ee{\end{equation}}
\def\bea{\begin{eqnarray}}
\def\eea{\end{eqnarray}}
\def\<{\langle}
\def\>{\rangle}

\def\nn{\nonumber}

\let\wtd=\widetilde
\def\Tr{{\rm Tr}}
\def\one#1{#1^{\raise5pt\hbox{$\scriptstyle\!\!\!\!1$}}\,{}}
\def\two#1{#1^{\raise5pt\hbox{$\scriptstyle\!\!\!\!2$}}\,{}}
\def\onetwo#1{#1^{\raise5pt\hbox{$\scriptstyle\!\!\!\!\!{12}$}}\,{}}
\def\otim{\mathop{\otimes}}
\def\tr{\mathop{\rm{tr}}}

\newtheorem{theorem}{Theorem}[section]
\newtheorem{lm}[theorem]{Lemma}
\newtheorem{prop}[theorem]{Proposition}
\theoremstyle{definition}
\newtheorem{df}[theorem]{Notation}

\newtheorem{remark}[theorem]{Remark}

\newtheorem{cor}[theorem]{Corollary}
\newtheorem{conjecture}[theorem]{Conjecture}

\allowdisplaybreaks

\begin{document}

\title[Poisson algebras  of b.u.t. bilinear forms and braid group]
{Poisson algebras of block-upper-triangular bilinear forms and braid group action}
\author{Leonid Chekhov$^{\ast,\star}$}\thanks{$^{\ast}$Steklov Mathematical Institute and  Laboratoire Poncelet,
Moscow, Russia}\thanks{$^{\star}$Concordia University, Montr\'eal, Canada. Email: chekhov@mi.ras.ru.}
\author{Marta Mazzocco$^\dagger$}\thanks{$^\dagger$School of Mathematica, Loughborough University, UK}

\maketitle

\begin{abstract}
In this paper we study a quadratic Poisson algebra structure on the
space of bilinear forms on $\mathbb C^{N}$ with the property that
for any $n,m\in\mathbb N$ such that $n m =N$, the restriction of the
Poisson algebra to the space of bilinear forms with block-upper-triangular
matrix composed from blocks of size $m\times m$ is Poisson.
We classify all central elements and characterise the Lie algebroid structure compatible with
the Poisson algebra.
We integrate this algebroid obtaining the corresponding  groupoid of morphisms of  block-upper-triangular
bilinear forms. The groupoid elements automatically  preserve the Poisson algebra. We then
obtain the braid group action on the Poisson algebra as elementary generators within the
groupoid. We discuss the affinisation and
quantisation of this Poisson algebra, showing that in the case  $m=1$ the quantum affine
algebra is the twisted $q$-Yangian for ${\mathfrak o}_n$ and for
$m=2$ is the  twisted $q$-Yangian for ${\mathfrak sp}_{2n}$. We
describe the quantum braid group action in these two
examples and conjecture the form of this action for any $m>2$.
\end{abstract}


\section{Introduction}

In this paper we consider bilinear forms on $\mathbb C^N$ defined by
$$
\langle x,y\rangle:= x^T A y,\qquad \forall\, x,\,y\in\mathbb C^N,\qquad A\in GL_N(\mathbb C).
$$
By {\it block-upper-triangular bilinear form}\/ we mean a bilinear
form such that the defining matrix $A$ is block--upper--triangular. In particular we use the following:

\begin{df}\label{def-block}
We let a block-upper-triangular (b.u.t.) matrix ${\mathbb A}$ to be an $(nm)\times (nm)$-matrix
composed from blocks $\mathbb A_{I,J}$, $I,J=1,\dots,n$, of size
$m\times m$ with the block-upper-triangular structure: we impose the
restrictions that $\mathbb A_{I,J}=0$ for $I>J$ and $\det\mathbb A_{I,I}=1$ for all $I=1,\dots,n$.
We denote by $\mathcal A_{n,m}\subset GL_{nm}$ the set  of all such block-upper-triangular
matrices.
\end{df}

In \cite{ChM} it was proved that for any number of blocks $n$ and
for any size of blocks $m$ the following brackets on the matrix
elements $a_{i,j}$ of $\mathbb A$:
\bea
\label{Poisson}
\{a_{i,j},a_{k,l}\}&=&\bigl({\rm sign}(j-l)+{\rm sign}(i-k)\bigr)a_{i,l}a_{k,j}+\\
&&
+\bigl({\rm sign}(j-k)+1\bigr)a_{j,l}a_{i,k}
+\bigl({\rm sign}(i-l)-1\bigr)a_{l,j}a_{k,i}\nn
\eea
define a Poisson bracket on $\mathcal A_{n,m}$.

Note that the brackets (\ref{Poisson}) depend neither on the size of
the $m\times m$ blocks nor on the number $n^2$ of blocks, so that
the  full space $GL_{N}(\mathbb C)$ of non-singular $N\times N$ matrices, $N=nm$, admits this Poisson algebra
(\ref{Poisson}). In Theorem \ref{th:Poisson} we show that the block-upper-triangular case is a Poisson
reduction of the full algebra in $(GL_{N}(\mathbb C),\{\cdot,\cdot\})$.

In the case of one-dimensional blocks (i.e. upper triangular
matrices with $1$ on the diagonal) this algebra reduces to the
Dubrovin--Ugaglia \cite{dubrovin,uga} bracket appearing in Frobenius
manifold theory and extensively studied by Bondal
\cite{Bondal,Bon1}. Its quantisation is also known as Nelson--Regge
algebra in $2+1$-dimensional quantum gravity \cite{NR,NRZ},  and as
Fock--Rosly bracket \cite{FR} in Chern--Simons theory. We expect
that for generic $m$ this algebra may have some interesting meaning
in these fields.

The affine version of the algebra (\ref{Poisson}) is defined in terms of the generating function:
\be
{\mathcal G}_{i,j}(\lambda):= G^{(0)}_{i,j}+\sum_{p=1}^\infty G_{i,j}^{(p)}\lambda^{-p},
\label{G-def}
\ee
where $G_{i,j}^{(p)}$ denotes the entry $i,j$ of the matrix
$G^{(p)}$ and we impose that $G^{(0)}:={\mathbb A}$, our
block-upper-triangular matrix, allowing  the matrices
$G^{(p)}$ to be arbitrary full-size matrices for $p>0$. The Poisson brackets
are postulated to be~\cite{ChM}
\begin{eqnarray}
\left\{{\mathcal G}_{i,j}(\lambda),{\mathcal G}_{k,l}(\mu)\right\}&=&
\left({\rm sign}(i-k)-\frac{\lambda+\mu}{\lambda-\mu}\right){\mathcal G}_{k,j}(\lambda) {\mathcal G}_{i,l}(\mu)+
\nn\\
&&
+\left({\rm sign}(j-l)+\frac{\lambda+\mu}{\lambda-\mu}\right) {\mathcal G}_{k,j}(\mu) {\mathcal G}_{i,l}(\lambda)+
\nn\\
&&
+\left({\rm sign}(j-k)-\frac{1+\lambda\mu}{1-\lambda\mu}\right)
{\mathcal G}_{i,k}(\lambda){\mathcal G}_{j,l}(\mu)+\nn\\
&&+
\left({\rm sign}(i-l)+\frac{1+\lambda\mu}{1-\lambda\mu}  \right){\mathcal G}_{l,j}(\lambda){\mathcal G}_{k,i}(\mu).
\label{eq:Yangian}
\end{eqnarray}
We call the index $p$ the level of the corresponding element; elements of ${\mathbb A}$ are then called
{\it zero--level elements.}\/ Analogously to the case of (\ref{Poisson}),
the algebra (\ref{eq:Yangian}) is Poisson {\it for any choice of the zero level $\mathbb A\in \mathcal A_{n,m}$,}\/
for any $n,m$ such that $nm=N$.

In our paper~\cite{ChM} we related this affine extension
(\ref{eq:Yangian}) in the case $m=1$  to the algebra ${\mathfrak D}_n$
of geodesic functions on an annulus with $n$ ${\mathbb Z}_2$
orbifold points and, simultaneously, to the algebra of monodromy
data of a  $n+1$-dimensional Frobenius manifold with one
non-semi-simple point.  Still in the case $m=1$ this affine
extension (\ref{eq:Yangian}) turns out to be the semi-classical
limit of the twisted $q$-Yangian for the ${\mathfrak o}_n$ algebra
introduced by Molev, Ragoucy, and Sorba~\cite{MRS}. A first
generalisation of the above algebras to block-upper-triangular
matrix case was constructed by Molev and Ragoucy in~\cite{MR}, where
they developed the twisted Yangian $Y'_q(\mathfrak{sp}_{2n})$ for the
${\mathfrak sp}_{2n}$ algebra. In this construction, the zero-level
algebra was block-lower-triangular (equivalent to
block-upper-triangular  by simple transposition) with $2\times 2$
blocks and with the restriction that each diagonal $2\times 2$-block
have the unit determinant. In the work by Molev and Ragoucy a full
description of the braid group action on $Y'_q(\mathfrak{sp}_{2n})$
was still missing and this was in part the trigger to the present
work.\footnote{We are particularly  grateful to Alexander Molev for
asking this question to us.}

Before explaining our results on the braid group action we need to
illustrate the ones on the central elements. We characterise all
central elements of the Poisson algebra (\ref{Poisson}) and of its
affine extension (\ref{eq:Yangian}). They are of two types:
polynomial central elements and rational central elements; together
they form a set of $n\bigl[\frac{m+1}{2}\bigr]+\bigl[\frac{nm}{2}\bigr]$ algebraically independent central
elements (here we let $[\cdot]$ denote the integer part of a number).

In Theorem \ref{th:central}, we prove that the {\it polynomial
central elements} for the  Poisson algebra (\ref{Poisson}) are given
by the coefficients of $\lambda^{-k}$, $k=0,1,\dots, \left[\frac{N+2}{2}\right]$, of the polynomial
$$
\det(\mathbb A +\lambda^{-1}\mathbb A^T),
$$
while for the affine Poisson algebra (\ref{eq:Yangian}) they are generated by the formal series
$$
\det\left(\mathcal G(\lambda)\right).
$$
The {\it rational central elements} are the same for both Poisson
algebras (\ref{Poisson}) and (\ref{eq:Yangian}). They are defined by
the {\it bottom--left minors} of the diagonal blocks of the zero
level matrix $\mathbb A$, i.e. let $\mathbb A\in\mathcal A_{n,m}$,
for each diagonal block $\mathbb A^{(I)}:=\mathbb A_{I,I}$,
$I=1,\dots,n$ take
$$
M^{(I)}_d:=\det\left(\begin{array}{ccc}
a^{(I)}_{m-d+1,1}&\dots&a^{(I)}_{m+d-1,d}\\
\vdots&\dots&\vdots\\
a^{(I)}_{m,1}&\dots&a^{(I)}_{m,d}\end{array}
\right),
$$
where $a^{(I)}_{i,j}$ denotes the $i,j$-th entry of  $\mathbb A^{(I)}:=\mathbb A_{I,I}$,
then in Theorem \ref{th:new-central} we prove that for every $d=0,\dots,\left[\frac{m}{2}\right]$ and $I=1,\dots,n$ the quantities
$$
b^{(I)}_d:= \frac{M^{(I)}_d}{M^{(I)}_{m-d}}
$$
are central elements for both Poisson algebras (\ref{Poisson}) and (\ref{eq:Yangian}).

Having characterised all central elements, we are ready to produce
the braid group action on $\mathcal A_{n,m}$. For this,  we follow
Bondal's approach  \cite{Bondal,Bon1} which consists in introducing
a suitable notion of groupoid of block-upper-triangular quadratic
bilinear forms in such a way that the Poisson bracket on the base
space $\mathcal A_{n,1}$ is given in terms of the anchor map
associated to the corresponding Lie algebroid. In Bondal's case,
namely when $m=1$,  the Lie algebroid is isomorphic to the Lie
algebroid on the cotangent bundle $T^\ast\mathcal A_{n,1}$. As soon
as $m>2$ this ceases to be true, making the integration of the Lie
algebroid rather tricky. We solve this problem in Section
\ref{se:algebroid} where we characterise this groupoid.\footnote{This groupoid
is the natural phase space of the Poisson sigma model with
target space $\mathcal A_{n,m}$ \cite{CF}.
In this case we expect to be able to integrate the constraint equation explicitly.}

The braid group generators are obtained as those elements in the
groupoid which swap the blocks and satisfy the braid group
relations.
To be more precise, the braid group acting on  $\mathcal A_{n,m}$ is $B_n$ in which each braid $\beta_{I,I+1}$, $I=1,\dots,n-1$
acts by changes of coordinates
on $\mathbb C^N$. This action can be presented in the adjoint matrix form (see formula (\ref{beta}) below)
$\beta_{I,I+1}[{\mathbb A}]=B_{I,I+1}{\mathbb A}B_{I,I+1}^T$
with the matrix $B_{I,I+1}$ having the block form (\ref{BII+1}).

The extended braid group transformations for the affine algebra
(\ref{eq:Yangian}) in the case where the zero-level  matrix
${\mathbb A}$ has the block-upper-triangular form is given by the
same $\beta_{I,I+1}[\mathcal G(\lambda)]=B_{I,I+1}\mathcal G(\lambda)B_{I,I+1}^T$
and we have one extra generator $\beta_{n,1}$
given by the formulas (\ref{R-n1-A-infty}) and (\ref{eq:bn1}).

Since the braid group elements belong to the groupoid, they preserve our algebras  (\ref{Poisson}) and (\ref{eq:Yangian}).

Finally we provide a quantisation of the affine algebra
(\ref{eq:Yangian}) in terms of quantum reflection equation  for any
$m$ and give formulae for the quantum braid group action in the case
$m=1$ and $m=2$. This leads to an interesting explicit relation between the
Lie algebroid of infinitesimal morphism of the b.u.t. algebra (\ref{Poisson}) and its $R$-matrix structure.

\vskip 2mm \noindent{\bf Acknowledgements.}  The authors are
specially grateful to Alexei Bondal, Alexander Molev, Stefan Kolb and Kirill Mackenzie for many
enlighting conversations.

The work of L.Ch. was supported in part by the Ministry of Education and Science of the Russian Federation (contract 02.740.11.0608),
by the Russian Foundation for Basic Research (Grant Nos. 10-01-92104-JF$\_$a, 11-01-00440-a), by the Grant of Supporting Leading
Scientific Schools of the Russian Federation NSh-8265.2010.1, and by the Program Mathematical Methods for Nonlinear Dynamics.

The work of M. Mazzocco was supported  by the EPSRC Advanced Research Fellowship
ARF EP/D071895/1.

\section{Poisson reductions of the algebras (\ref{Poisson})  and  (\ref{eq:Yangian})}\label{se:red}

\begin{theorem}\label{th:Poisson}
The affine algebra (\ref{eq:Yangian}) is an abstract infinite-dimensional Poisson algebra for  any choice of the block-upper-triangular
form of the zero level  matrix ${\mathbb A}\in\mathcal A_{n,m}$.
Analogously, the restriction of the brackets (\ref{Poisson}) on  $GL_{N}(\mathbb C)$ to the
block-upper-triangular matrices ${\mathbb A}_{n,m}$ for any $n,m\in\mathbb N$ such that $nm=N$, is Poisson.\end{theorem}

\proof
The proof of the Jacobi relations in Appendix~A of~\cite{ChM}
used only combinatorial properties and was independent on possible reductions.
So, it remains only to prove the consistency of the reductions with the affine algebra (the consistency
of the reductions of  (\ref{Poisson})  is a trivial corollary). For this,
let us calculate the bracket between elements of the zeroth and $k>0$ levels. From (\ref{eq:Yangian}), we have
(one can obtain the formula below by taking a formal limit $\lambda\to\infty$)
\bea
\{a_{i,j},G^{(p)}_{k,l}\}&=&
({\rm sign}(i-k)-1)a_{k,j}G^{(p)}_{i,l}
+({\rm sign}(j-l)+1) G^{(p)}_{k,j} a_{i,l}
\nn\\
&&
+({\rm sign}(j-k)+1)a_{i,k}G^{(p)}_{j,l}
+({\rm sign}(i-l)-1)a_{l,j}G^{(p)}_{k,i}.
\label{eq:zero-p}
\eea
The right-hand side is nonzero (due to combinations of ${\rm sign}$-factors) only for $i\leq k$ and/or $l\leq j$
and/or $k\leq j$ and/or $i\leq i$. We now use the specific form of the reduction, namely the fact that if we
impose $a_{i,j}=0$, then $a_{s,j}=0$ for all $s\geq i$ and $a_{i,t}=0$ for all $t\leq j$.
Therefore if $i\leq k$ then $a_{k,j}$ is zero and the term $({\rm sign}(i-k)-1)a_{k,j}G^{(p)}_{i,l}$ does not contribute.
Analogously, for  $k\leq j$ we have that $a_{i,k}=0$ and the term $({\rm sign}(j-k)+1)a_{i,k}G^{(p)}_{j,l}$ does not contribute.
The same happens if  $l\leq j$ and/or  $i\leq i$. This proves the consistency between our reduction and the algebra (\ref{eq:Yangian}).
\endproof

\begin{remark}\label{rem:squares}
A more general statement is true: let us consider block-upper-triangular matrices with blocks of different sizes, or in other words let us consider an arbitrary partition of $N$ (previously equal to $mn$) into
$n$ positive integers, $N=m_1+\cdots+m_n$, and let $\mathbb A_{I,J}$ be a matrix of size $m_I\times m_J$.
All the constructions of this paper, including the Poisson restriction (Theorem~\ref{th:Poisson}), central elements, and the action of the
groupoid of b.u.t. matrices can be straightforwardly generalised to this case as well except the (classical and quantum) braid-group action,
which is apparently lost in the case of different block sizes.

If we consider even more general case in which we loose the block upper triangular form,  and take the Poisson reduction depicted in Fig.~\ref{fig:Young} where
all elements below a broken line that goes as in the figure are set to be zeros, then Theorem~\ref{th:Poisson} still remains
valid, but we no longer have an algebra structure as the product
of two matrices of this form does not have the same form.
\end{remark}

\begin{figure}[tb]
{\psset{unit=0.7}
\begin{pspicture}(-2.5,-2.5)(2.5,2.5)
\psframe[linecolor=yellow, fillstyle=solid, fillcolor=yellow](-2.5,1)(2.5,2.5)
\psframe[linecolor=yellow, fillstyle=solid, fillcolor=yellow](-1.5,0)(2.5,1)
\psframe[linecolor=yellow, fillstyle=solid, fillcolor=yellow](0.5,-1.5)(2.5,0)
\psframe[linecolor=yellow, fillstyle=solid, fillcolor=yellow](2,-2.5)(2.5,-1.5)
\psframe[linecolor=red, fillstyle=solid, fillcolor=red](-2.5,1)(-2.2,1.3)
\psframe[linecolor=red, fillstyle=solid, fillcolor=red](-1.5,0)(-1.2,.3)
\psframe[linecolor=red, fillstyle=solid, fillcolor=red](.5,-1.5)(.8,-1.2)
\psframe[linecolor=red, fillstyle=solid, fillcolor=red](2,-2.5)(2.3,-2.2)
\psframe[linewidth=1pt](-2.5,-2.5)(2.5,2.5)
\pcline[linecolor=blue, linestyle=dashed, linewidth=1.5pt](-2.5,1)(-1.5,1)
\pcline[linecolor=blue, linestyle=dashed, linewidth=1.5pt](-1.5,1)(-1.5,0)
\pcline[linecolor=blue, linestyle=dashed, linewidth=1.5pt](-1.5,0)(.5,0)
\pcline[linecolor=blue, linestyle=dashed, linewidth=1.5pt](.5,0)(.5,-1.5)
\pcline[linecolor=blue, linestyle=dashed, linewidth=1.5pt](.5,-1.5)(2,-1.5)
\pcline[linecolor=blue, linestyle=dashed, linewidth=1.5pt](2,-1.5)(2,-2.5)
\psellipse[linewidth=1.5pt](-1,-1.3)(0.3,0.5)
\end{pspicture}
}
\caption{A general Poisson reduction of the algebra (\ref{Poisson}). All the items
below the dashed broken line are zeros. The pivotal elements at the corners are marked by
dark squares.}
\label{fig:Young}
\end{figure}

\subsection{Reduction to the symmetric matrices}
The Poisson structure (\ref{Poisson}) restricts also to the space $Sym_{N}$ of symmetric $N\times N$ matrices.

\begin{prop}
The restriction of the  Poisson structure (\ref{Poisson}) to the space $Sym_{N}$ of symmetric matrices is Poissonian.
\end{prop}

\begin{proof}
Let us consider the Poisson bracket of the combination $a_{i,j}-a_{j,i}$ with any element $a_{k,l}$;
\bea
\left\{a_{i,j}-a_{j,i},a_{k,l}\right\}&=&{\rm sign}(j-l)a_{k,j}(a_{i,l}-a_{l,i}) +
{\rm sign}(i-k)a_{i,l}(a_{k,j}-a_{j,k}) +\nn\\
&&+
{\rm sign}(j-k)a_{j,l}(a_{i,k}-a_{k,i}) +
{\rm sign}(i-l)a_{k,i}(a_{l.j}-a_{j,l})+\nn\\
&&+a_{j,l} a_{i,k}-a_{l,j} a_{k,i}
-a_{i,l} a_{j,k}+a_{l,i} a_{k,j}. \nn
\eea
By imposing the condition $a_{r,s}=a_{s,r}$ for all $r,s$, the above expression is always $0$.
\end{proof}

The reduced bracket on $Sym_{N}$  reads:
\be\label{eq:sym}
\left\{a_{i,j},a_{k,l}\right\}=\left({\rm sign}(j-l)+{\rm sign}(i-k)\right)a_{i,l}a_{k,j}+
\left({\rm sign}(j-k)+{\rm sign}(i-l)\right)a_{j,l}a_{k,i}
\ee
This Poisson structure on $Sym_{N}$  was already studied by Bondal \cite{Bon1}.

\begin{remark}
Observe that on the contrary the affine algebra (\ref{eq:Yangian}) is not compatible with the restriction $\mathbb A\in Sym_{N}$.
This can be easily seen by  observing that
$\{a_{i,j}-a_{j,i},G^{(p)}_{k,l}\}\neq 0$ for  $\mathbb A\in Sym_{N}$. We do not know whether an affine extension of (\ref{eq:sym}) exists.
\end{remark}

\subsection{$k$-level reductions of the twisted Yangian extension}

\begin{df}\label{k-level-reduction}
We call the $k$-level reduction of the algebra (\ref{eq:Yangian}) the mapping
\be
\label{k-level}
{\mathcal G}(\lambda)\mapsto {\mathbb A}+\lambda^{-1}\widehat{ G}^{(1)}+\cdots+\lambda^{-k+1}\widehat{ G}^{(k-1)}
+\lambda^{-k}{\mathbb A}^T,
\ee
where
\be
\label{k-level-sym}
\widehat{ G}^{(i)}=\left[\widehat{ G}^{(k-i)}\right]^T
\ee
and $\widehat{ G}^{(i)}={ G}^{(i)}$
for $i=1,\dots, (k-1)/2$ for odd $k$ and for $i=1,\dots, k/2-1$ for even $k$ and $\widehat{ G}^{(k/2)}_{i,j}={ G}^{(k/2)}_{i,j}$
for $i\ge j$ for even $k$, whereas the other entries of $\widehat{ G}^{(i)}$ for $i\ge k/2$ are defined by the symmetry condition
(\ref{k-level-sym}).
\end{df}

\begin{theorem}\label{th:k-level}
The mapping (\ref{k-level}) defines a surjective homomorphism of the algebra (\ref{eq:Yangian})
to the corresponding algebra of the elements ${\mathbb A}_{i,j}$ and $\widehat{ G}^{(l)}_{i,j}$, $l=1,\dots,k-1$, for any $k\in {\mathbb Z}_+$.
\end{theorem}

\section{Central elements}\label{central}

In this section, we construct all the central elements of the
algebra of block-upper-triangular matrices (\ref{Poisson}) and of
its twisted Yangian extension (\ref{eq:Yangian}). In order to  find
them we first need to characterise a simple automorphism and a
simple anti-automorphism of the Poisson algebra (\ref{Poisson}).

\subsection{(Anti)automorphisms of the Poisson algebra}\label{s:antiautomorphism}

Let
$N=nm$ denote the total size of the matrix ${\mathbb A}$. Then the transformation
\be
\label{reflection}
P[{\mathbb A}]=\wtd{\mathbb A},\quad {\wtd a}_{i,j}=a_{N+1-j,N+1-i}
\ee
is an antiautomorphism of the Poisson algebra (\ref{Poisson}), that is,
\be
\label{substitution}
\bigl\{{\wtd a}_{i,j},{\wtd a}_{k,l}\bigr\}=-\Bigl.\bigl\{{ a}_{i,j},{ a}_{k,l}\bigr\}\Bigr|_{a\mapsto {\wtd a}}.
\ee

Besides it we have the scaling transformation, which obviously leaves invariant the algebra (\ref{Poisson}):
\be
\label{scaling}
a_{i,j}\mapsto e^{\phi_i+\phi_j}a_{i,j},\qquad \phi_i=\phi_{N+1-i},
\ee
where we impose the restriction on $\phi_i$ in order to ensure the transformation
(\ref{scaling}) to be consistent with the antiautomorphism (\ref{reflection}). We also impose that
$\sum_{i=Jm+1}^{Jm+m}\phi_i=0$, $J=1,\dots,n$, to ensure the preservation of the determinant condition
$\det A_{J,J}=1$ for any $J$.

\begin{remark}
Note that the fact that the scaling transformation (\ref{scaling})
is an automorphism of the algebra (\ref{Poisson}) allows to restrict
this algebra on the projective space $\mathbb P^{N^2-1}$.
This fact is relevant due to the recent interest in the vanishing locus
of quadratic Poisson algebras on projective spaces in algebraic geometry \cite{hit}.
\end{remark}

\subsection{Polynomial central elements}

In this subsection, we construct a part of central elements that can be obtained by standard methods based on
algebra symmetries as, say, in \cite{Bondal} or \cite{Molev}.

\begin{theorem}\label{th:central}
The polynomial functions of the elements of the algebra (\ref{eq:Yangian}) in the infinite-series
expansion of $\det {\mathcal G}(\mu)$ in powers of $\mu^{-1}$ are central elements of the affine algebra
(\ref{eq:Yangian}).
\end{theorem}

\proof
Although it follows from the more general statement by Molev and Ragoucy~\cite{MR} on the central elements of the (quantum) Yangian
algebra, we can easily verify it directly using that
$$
\{{\mathcal G}_{i,j}(\lambda),\det{\mathcal G}(\mu)\}=\sum_{k,l=1}^{nm}
\{{\mathcal G}_{i,j}(\lambda),{\mathcal G}_{k,l}(\mu)\}{\mathcal G}^{-1}_{l,k}(\mu)
$$
(the invertibility of ${\mathbb A}$ ensures the existence of the inverse matrix ${\mathcal G}^{-1}(\mu)$ at least as a formal series).
We now substitute the bracket (\ref{eq:Yangian}), and using the obvious identities
$\sum_{l=1}^{nm}{\mathcal G}_{x,l}(\mu){\mathcal G}^{-1}_{l,k}(\mu)=\delta_{x,k}$ and
$\sum_{k=1}^{nm}{\mathcal G}_{k,x}(\mu){\mathcal G}^{-1}_{l,k}(\mu)=\delta_{l,x}$ for $x=i,j$, we obtain zero.
\endproof

\begin{cor}\label{cor:central}
The coefficients in the $\lambda^{-1}$-expansion of $\det({\mathbb
A}+\lambda^{-1}{\mathbb A}^T)$ are central elements of the Poisson
algebra (\ref{Poisson}) restricted to the block-upper-triangular
matrices ${\mathbb A}\in\mathcal A_{n,m}$ for any choice of
$n,m\in\mathbb N$ such that $nm=N$. They form a family of
$\left[\frac{N}{2}\right]$ algebraically independent central
elements.
\end{cor}

\proof
We need the statement of Theorem~\ref{th:k-level} for $k=1$:

\begin{lm}\label{lm:homomorphism}
The mapping
\be
\label{eq:homom}
{\mathcal G}(\lambda)\mapsto {\mathbb A}+\lambda^{-1}{\mathbb A}^T
\ee
defines a surjective homomorphism of the algebra (\ref{eq:Yangian}) to the algebra (\ref{Poisson}).
\end{lm}

\proof The  proof of this Lemma is obtained by a direct substitution of
expression (\ref{eq:homom}) into (\ref{eq:Yangian}) using the algebra (\ref{Poisson}).\endproof

\noindent {\em Proof of Corollary \ref{cor:central}.} The proof  that the coefficients in the $\lambda^{-1}$-expansion of
$\det({\mathbb A}+\lambda^{-1}{\mathbb A}^T)$
are central elements of the Poisson algebra (\ref{Poisson}) follows directly from Theorem \ref{th:central}.
The fact that no more than  $\left[\frac{N+2}{2}\right]$  of them are algebraically independent follows from the simple observation that
$$
\det({\mathbb A}+\lambda^{-1}{\mathbb A}^T)=\frac{c_0\lambda^N + c_1\lambda^{N-1}+\dots+c_N }{\lambda^N},
$$
where $c_{N-k}=c_k$ for all $k=0,1,\dots, \left[\frac{N}{2}\right]$ and
$c_0=\det(\mathbb A_{11})\det(\mathbb A_{22})\cdot\dots\cdot\det(\mathbb A_{nn})=1$.

The fact that generically the coefficients $c_1,\dots,c_{ \left[\frac{N}{2}\right]}$ form a
family of $\left[\frac{N}{2}\right]$ algebraically independent central elements
was proved in \cite{dubrovin} for the most reduced case $m=1$.
\endproof

\subsection{Rational central elements}

The central elements in Corollary~\ref{cor:central}
do not  exhaust all the central elements of the algebra of entries of ${\mathbb A}$.
We also have rational central elements. To describe them we begin by considering the case of
the {\it  nonrestricted Poisson algebra} $(GL_N,\{\cdot,\cdot\})$ where $\{\cdot,\cdot\}$ is given by (\ref{Poisson}),  and make the following

\vskip 1.5mm
\noindent {\bf Generality assumption:} All the minors $M_d$ of size $d\times d$ located at the lower-left corner are non-zero.

\begin{theorem}\label{central-A}
Under the above generality assumption, the elements
$$
\det M_{N-d}/\det M_d, \quad\hbox{for}\quad d=0,\dots,\left[\frac{N-1}{2}\right],
$$
are central for the affine algebra (\ref{eq:Yangian}) and are algebraically independent in the nonrestricted case.
\end{theorem}

\proof
The proof is based on the following:

\begin{lm}\label{lm:commutation}
For the nonrestricted $N\times N$ matrix ${\mathbb A}$ in our genericity assumption,
denoting $a_{k,l}=G^{(0)}_{k,l}$,
we have the following commutation relations:
\be\label{Md-aIJ}
\{\det M_d,G^{(p)}_{k,l}\}=c^d_{k,l} G^{(p)}_{k,l}\det M_d \ \hbox{for} \ p=0,1,\dots,
\ee
where
\be
\label{cdkl}
c^d_{k,l}=-\delta_{k+d>N}+\delta_{d+1>l}+\delta_{d+1>k}-\delta_{l+d>N}
\ee
where  $\delta_{i>j}=1$ for $i>j$ and $0$ otherwise.
\end{lm}

\proof Let us deal with the minor $M_2$ first. By
the Leibnitz rule, we obtain four brackets,
and using relation (\ref{eq:zero-p}) we obtain four terms for each bracket.
Grouping  together terms with the same $G_{r,s}^{(p)}$ entry we obtain
\bea
\left\{M_d,G_{k,l}^{(p)}\right\}&=&({\rm sign}(N-k)-1)G_{N,l}^{(p)} \left|\begin{array}{cc}a_{N-1,1}&a_{N-1,2}\\
a_{k,1}&a_{k,2}\\
\end{array}\right| + \nn\\
&+&
({\rm sign}(2-l)+1)G_{k,2}^{(p)} \left|\begin{array}{cc}a_{N-1,1}&a_{N-1,l}\\
a_{N1}&a_{N,l}\\
\end{array}\right| \nn\\
&+&({\rm sign}(2-k)+1)G_{2,l}^{(p)} \left|\begin{array}{cc}a_{N-1,1}&a_{N-1,k}\\
a_{N,1}&a_{N,k}\\
\end{array}\right| + \nn\\
&+&
({\rm sign}(N-l)-1)G_{k,N}^{(p)} \left|\begin{array}{cc}a_{N-1,1}&a_{N-1,2}\\
a_{l,1}&a_{l,2}\\
\end{array}\right|+ \\
&+&({\rm sign}(N-1-k)-1)G_{N-1,l}^{(p)} \left|\begin{array}{cc}a_{k,1}&a_{k,2}\\
a_{N,1}&a_{N,2}\\
\end{array}\right| + \nn\\
&+&
({\rm sign}(1-l)+1)G_{k,1}^{(p)} \left|\begin{array}{cc}a_{N-1,l}&a_{N-1,2}\\
a_{N,l}&a_{N,2}\\ \end{array}\right| + \nn\\
&+&({\rm sign}(1-k)+1)G_{1,l}^{(p)} \left|\begin{array}{cc}a_{N-1,k}&a_{N-1,2}\\
a_{N,k}&a_{N,2}\\
\end{array}\right| + \nn\\
&+&
({\rm sign}(N-1-l)-1)G_{k,N-1}^{(p)} \left|\begin{array}{cc}a_{l,1}&a_{l,2}\\
a_{N1}&a_{N,2}\\
\end{array}\right|
\nn
\eea
and each term in this sum is nonzero only for one choice of either $k$ or $l$.
For example consider the last term on the r.h.s.:
the coefficient $({\rm sign}(N-1-l)-1)$ is nonzero only for $l=N-1$ or $l=N$.
However in the latter case the determinant $ \left|\begin{array}{cc}a_{l,1}&a_{l,2}\\
a_{N1}&a_{N,2}\\
\end{array}\right| $ becomes zero, so we may only choose  $l=N-1$. It easily follows that (\ref{Md-aIJ}) and (\ref{cdkl}) are satisfied.

In the case of $d>2$ the computation is very similar: the first and fifth term above are replaced by the
 sum of $d$ determinants enumerated by the index $i=N-d+1,\dots,N$ and such that in each
of the corresponding matrices the $i$th row vector is replaced by the $k$th row vector multiplied by $({\rm sign}(i-k)-1)G^{(p)}_{i,l}$.
If $i<k$, the corresponding determinant is zero (the matrix then contains two proportional row vectors), so
the only nonzero contribution occurs when $i=k$, which is possible only if $k>n-d$, and this contribution is $-G^{(p)}_{k,l}\det M_d$, which
gives the first term in the r.h.s. of (\ref{cdkl}). Using the same reasonings we can deal with three other terms.
Because $\delta_{i<0}=1-\delta_{i+1>0}$
we easily obtain from (\ref{cdkl}) that $c^d_{k,l}=c^{n-d}_{k,l}$, which completes the proof of the Lemma.
\endproof

The proof of the fact that the elements $\det M_{N-d}/\det M_d$ for $d=0,\dots,\left[\frac{N-1}{2}\right]$ are central then
follows by the Leibnitz rule for the Poisson bracket and by observing that
$c_{k,l}^d=c_{k,l}^{N-d}$ for all $k,l,d$, so that  $\det M_d$ and $\det M_{N-d}$ have {\it exactly
the same} commutation relations with all of $a_{k,l}$ and with all of $g^{(p)}_{k,l}$ for $p\ge 1$ in the twisted Yangian case.

That these central elements are algebraically independent was proved by Bondal \cite{Bon1} already for the restriction of the algebra (\ref{Poisson})
to $Sym_N$.
\endproof

We now formulate the general algebraic independence lemma valid in the nonrestricted case.
\begin{lm}\label{lemma-alg-ind}
The set of algebraically independent central elements of the nonrestricted
Poisson algebra $(GL_N,\{\cdot,\cdot\})$ where $\{\cdot,\cdot\}$ is given by (\ref{Poisson}),
comprises the coefficients $c_k$ of $\lambda^{-k}$, $k=0,1,\dots,\left[N/2\right]$, of the expansion of $\det({\mathbb A}+\lambda^{-1}{\mathbb A}^T)$
and the rational central elements $b_l=\det M_{N-l}/\det M_l$, $l=1,\dots,\left[(N-1)/2\right]$ provided all $\det M_l$,
$l=1,\dots,\left[(N-1)/2\right]$ are nonzero.
\end{lm}

\proof
We have already proved that these elements are central and that each set $\{c_k\}$ and $\{b_l\}$ is algebraically independent. Suppose we have a
function
$$
F\bigl( \{c_k\}, \{b_l\}\bigr)=0 \hbox{ for all values of }a_{i,j}.
$$
Because any transformation (\ref{scaling}) is an automorphism of the
algebra (\ref{Poisson}), choosing $\phi_l^{(i)}=s_i(\delta_{l,i+1}-\delta_{l,i})$ for $i=1,\dots,\left[(N-1)/2\right]$ we obtain that all $c_k$ and
all $b_l$ with $l\ne i$ remain invariant whereas $b_i\to b_ie^{2s_i}$. This means that if $F\bigl( \{c_k\}, \{b_l\}\bigr)=0$ for some nonzero
$\{b_l\}$, then $F\bigl( \{c_k\}, \{b_le^{2\phi_l}\}\bigr)=0$ for any choice of $\phi_l\in{\mathbb R}$. Hence, the function $F$ is actually
independent of all of $b_l$ and we have that $F(\{c_k\})$=0. Because the set of $\{c_k\}$ is algebraically independent, we have that $F\equiv 0$,
which completes the proof.
\endproof

Adding $\det{\mathbb A}$, which
corresponds both to the rational central element with $d=0$ and to the polynomial central element given by the coefficient
of power $0$, to the set we have $\left[\frac{N+1}{2}\right]$ central elements described by Theorem~\ref{central-A}
and $\left[\frac{N}{2}\right]$ central elements from Corollary~\ref{cor:central}, so, in the general case of
a nonrestricted algebra $(GL_N,\{\cdot,\cdot\})$ where $\{\cdot,\cdot\}$ is given by (\ref{Poisson}),
we have exactly $N$ algebraically independent central elements.

\begin{remark}\label{higher-dimension}
Elementary, but lengthy calculations demonstrate that the highest Poisson
leaf dimension is not less than $N^2-N$.
Here we only briefly outline the way of proving it. For this it suffices to
consider the case where all $a_{i,j}$ with $i\ne j$ are $\epsilon$-small
as compared to all $a_{i,i}$ and to retain only terms of
order $O({\epsilon})$ in the Poisson relations (\ref{Poisson}) neglecting
all the terms of order $\epsilon^2$.
Introducing then the combination $b_{i,j}=a_{i,j}-a_{j,i}$ and retaining
the elements $a_{i\ge j}$ with $i\ge j$ we observe that
in the limit of small $\epsilon$, all the $b_{i,j}$ commute with all the
$a_{i\ge j}$, so that the Poisson algebra splits in two sub-algebras, the $a_{i\ge j}$-algebra and the
$b_{i,j}$-algebra.

The $b_{i,j}$-algebra becomes  the small-$\epsilon$ limit of the Dubrovin--Ugaglia or Nelson--Regge algebra (\ref{eq:du})
and therefore its  highest Poisson leaf dimension is $\frac{N(N-1)}{2}-\left[\frac{N}{2}\right]$. The
$a_{i\ge j}$-algebra becomes the Dubrovin--Ugaglia or Nelson--Regge algebra
to which we add the diagonal elements, and therefore its  highest Poisson leaf dimension is
$\frac{N(N+1)}{2}-\left[\frac{N+1}{2}\right]$, so the highest rank of the
Poisson relations (\ref{Poisson}) will be not less than
$N^2-N$, as expected.\end{remark}

We are now going to formulate the theorem describing the rational central elements in the general case of the block-upper-triangular matrix
${\mathbb A}$ and its possible Yangian extensions. We begin by fixing our notation.
\begin{df}\label{minorsMdi}
For a block-upper-triangular matrix ${\mathbb A}$ with $n$ blocks of sizes $m_i\times m_i$, $i=1,\dots,n$ on the diagonal, let
$M_d^{(i)}$, $d=0,\dots, m_i$, $i=1,\dots,n$, be minors of size $d\times d$ located at lower-left corners of the corresponding
{\em diagonal blocks} of the matrix ${\mathbb A}$.
\end{df}

\begin{theorem}\label{th:new-central}
Provided $\det M_d^{(i)}$ are nonzero, all the quotients
$$
b_d^{(i)}\equiv\det M_{m_i-d}^{(i)}/\det M_d^{(i)},\qquad d=0,\dots,[(m_i-1)/2], \quad i=1,\dots,n
$$
are central elements of both the algebra (\ref{Poisson}) and  its Yangian extension
(\ref{G-def})--(\ref{eq:Yangian}), for any choice of the zero level
 ${\mathbb A}\in\mathcal A_{n,m}$.

The central elements $b_d^{(i)}$, $d=0,\dots,[(m_i-1)/2]$,  $i=1,\dots,n$, and the coefficients $c_r$ of $\lambda^{-r}$
terms ($r=1,2,\dots$) of the expansion of $\det{\mathcal G}(\lambda)$ constitute an algebraically independent complete set of central elements
of the affine algebra (\ref{eq:Yangian}) whose zero level ${\mathbb A}$
is restricted to the block-upper-triangular form ${\mathbb A}\in\mathcal A_{n,m}$ for any choice of $n,m$.

These central elements remain central for all the $k$-level reductions (\ref{k-level}). In this case,
the complete set of algebraically independent central elements comprises the same elements $b_d^{(i)}$ as above
and the elements $c_r$ with $r=1,\dots, [(Nk)/2]$.

In particular, the maximal dimension of the Poisson leaves for the algebra  (\ref{Poisson}) on
 $\mathcal A_{n,m}$ is
 $$
\frac{n(n+1)}{2} m^2- n m - s\left[\frac{n}{2}\right],\quad s=\left\{\begin{array}{l}
 1\quad\hbox{for $m$ odd,}\\
 0\quad\hbox{for $m$ even}
 \end{array}\right.
$$
which is always even.
\end{theorem}

\proof The proof of the fact that the quotients $b_d^{(i)}$ for $ d=0,\dots,[(m_i-1)/2]$, and $i=1,\dots,n$
are central elements is analogous to the proof of Theorem~\ref{central-A} with the only distinction that
now some of the row or column vectors will be zero because of the Poissonian restrictions. The proof of algebraic independence is
 analogous to the proof of Lemma~\ref{lemma-alg-ind} in which we must generalise the automorphism (\ref{scaling}) to the
affine case by setting $G^{(p)}_{i,j}\mapsto G^{(p)}_{i,j}e^{\phi_i+\phi_j}$ for all $p=0,1,\dots$.
The computation of the maximal dimension of the Poisson leaves  for the algebra  (\ref{Poisson}) on
 $\mathcal A_{n,m}$ follows from the fact that the affine space  $\mathcal A_{n,m}$ has dimension
 $\frac{n(n+1)}{2} m^2-n$ because we have $\frac{n(n-1)}{2}$ off diagonal blocks with $m^2$ elements each,
 and $n$ diagonal blocks with $m^2-1$ elements each.
 We then need to subtract from this the number of algebraically central elements. These are $b_d^{(i)}$ for
  $d=1,\dots,[(m-1)/2]$, and $i=1,\dots,n$, giving $\left[\frac{m-1}{2}\right] n$ algebraically
  independent central elements, and $c_k$, $k=1,\dots, \left[\frac{m n}{2}\right]$, giving another $ \left[\frac{m n}{2}\right]$
  algebraically independent central elements. \endproof

\begin{remark}\label{rem:symmetries}
Note that the constructed central elements are of two, very different, sorts. Those generated by $\det({\mathbb A}+\lambda^{-1}{\mathbb A}^T)$
are invariant under the transformation (\ref{scaling}) whereas, providing all $\det M_d^{(i)}$ are nonzero, we can use
transformations (\ref{scaling}) to set all the central elements $\det M_{m_i-d}^{(i)}/\det M_d^{(i)}$ equal to $\pm 1$ (in the case
of real parameters $\phi_s$). Then, the group of (anti)automorphisms of the Poisson algebra (\ref{Poisson}) in the case of the
block-upper-triangular matrices from Definition~\ref{def-block} is presumably generated by the braid group transformations (\ref{beta})
with $I=1,\dots,n-1$, by the anti-automorphism $P$ from (\ref{reflection}), and, possibly, by ``inner'' automorphisms $\beta_i$
(in terminology of Molev and Ragoucy paper \cite{MR}, where these
automorphisms were constructed for
the case $m=2$), $i=1,\dots,n$, that act nontivially only inside the blocks ${\mathbb A}_{i,i}$, ${\mathbb A}_{i,k}$ ($n\ge k>i$),
and  ${\mathbb A}_{l,i}$ ($i>l\ge 1$) and mutually commute.
In the case of general $m$, these automorphisms, if exist, must also commute mutually, and to preserve
the rational central elements, we expect they to preserve the structure of minors $M_d^{(i)}$, that is, they must correspond
to a sort of transposition w.r.t. the antidiagonal of the matrix ${\mathbb A}_{i,i}$. The
problem of existence of these automorphisms is under investigation.
\end{remark}

\section{Groupoid of block upper triangular bilinear forms}\label{se:algebroid}

In this section, we follow Bondal's approach  \cite{Bondal,Bon1}
which consists in introducing a suitable notion of groupoid of
block--upper--triangular quadratic bilinear forms in such a way that
the Poisson bracket on the base space $\mathcal A_{n,1}$ is given in
terms of the anchor map associated to the corresponding Lie
algebroid. In this approach the braid group elements are then
obtained as elementary generators of this groupoid and automatically preserve the Poisson structure.

In Bondal's case, namely when $m=1$,  the Lie algebroid
is isomorphic to the Lie algebroid on the cotangent bundle
$T^\ast\mathcal A_{n,1}$. As soon as $m>2$ this ceases to be true,
making the integration of the Lie algebroid rather tricky. We solve
this problem in this Section. Let us first recall Bondal's
construction.

\subsection{The case of upper-triangular matrices with one on the diagonal}\label{se:bondal}

In this section we recall some key results from Bondal's work   \cite{Bondal,Bon1}, or at least our interpretation of them.

Denote by $\mathcal A\subset GL_{n}(\mathbb C)$ the set  of all
upper--triangular matrices $\mathbb A$ with $1$ on the diagonal.
The Lie group $GL_{n}(\mathbb C)$ acts on
$\mathbb C^n$ in the usual way, thus acting on bilinear forms as
$$
\forall A, B\in GL_n(\mathbb C),\qquad
A\mapsto B A B^T.
$$
This action of $GL(\mathbb C^{n})$ does not preserve $\mathcal A$,
however, for any element ${\mathbb A}\in\mathcal A$, one can take the
subset $\mathcal M_{\mathbb A}\subset GL(\mathbb C^{n})$ of
elements that preserve the structure of ${\mathbb A}$, or in other
words
\be
\label{groupoid}
\mathcal M_{\mathbb A}=\left\{B\in GL(\mathbb C^{n})\, | \,
{\mathbb A}\mapsto B{\mathbb A}B^T\in \mathcal A
\right\}.
\ee
Let  $(\mathcal A,\mathcal M)$ where $\mathcal M=\cup_{\mathbb
A\in \mathcal A} \mathcal M_{\mathbb A}$ be the set of pairs
$(\mathbb A,B)$ such that $\mathbb A\in {\mathcal A}$ and
$B\in\mathcal M_{\mathbb A}$. The identity morphism is defined as
\be\label{eq:emap}
e=(\mathbb A,\ID),
\ee
the inverse as
\be\label{eq:imap}
i: (\mathbb A,B)\to (B \mathbb A B^{T},B^{-1}),
\ee
and the partial multiplication as
\be\label{eq:mmap}
m\left((B_1 \mathbb A B_1^T,B_2),(\mathbb A, B_1)\right)= (\mathbb A, B_2 B_1).
\ee
These rules define the structure of smooth algebraic groupoid on $(\mathcal A,\mathcal M)$
\cite{Bondal}. A smooth groupoid naturally defines a Lie algebroid $(\mathcal A,\mathfrak g)$, i.e.
its infinitesimal version:
$$
\mathfrak g:= \cup_{\mathbb A\in\mathcal A}\mathfrak g_{\mathbb A}
$$
where
$$
\mathfrak{g}_{\mathbb A}:=\left\{ g\in\mathfrak{gl}_{n}(\mathbb C),|\,
\mathbb A+ {\mathbb A}g+g^T{\mathbb A}\in\mathcal A \right\}.
$$
We denote by $D_{\mathbb A}$ the anchor map
\be
\label{gA}
\begin{array}{lccc}
D_{\mathbb A}:&\mathfrak g_{\mathbb A}&\to& T_{\mathbb A}\mathcal A\\
& g &\mapsto &{\mathbb A}g+g^T{\mathbb A}.\\ \end{array}
\ee
The Lie bracket on the space of sections $\Gamma(\mathfrak g)$ is defined by
\be\label{eq:lieb}
\left[v_1,v_2\right]_{\Gamma}(\mathbb A):=[g_1,g_2] + \sum_{i,j} \frac{\partial v_2}{\partial a_{i,j}}
\left(D_{\mathbb A}(g_1)\right)_{i,j} - \frac{\partial v_1}{\partial a_{i,j}} \left(D_{\mathbb A}(g_2)\right)_{i,j} ,
\ee
where for $i=1,2$ $v_i\in\Gamma(\mathfrak g)$ and we denoted by
$g_i\in\mathfrak g_{\mathbb A}$ the image of $\mathbb A\in\mathcal A$ under $v_i$. Here the first term in the right hand side is the
usual matrix commutator.

The following Lemma is based on the fact that the tangent bundle
$T\mathcal A$ can be identified with the space of strictly upper
triangular matrices, while the cotangent  bundle $T^\ast\mathcal A$
can be identified with the space of strictly lower triangular
matrices by the Killing form, which is given simply by the trace in
this case.

\begin{lm}\label{lemma-g}\cite{Bondal}
The map
\be
\label{P_A}
\begin{array}{lccl}
P_{\mathbb A}:&T^\ast_{\mathbb A}\mathcal A&\to& \mathfrak g_{\mathbb A}\\
& w &\mapsto &P_{-,1/2}(w{\mathbb A})-P_{+,1/2}(w^T{\mathbb A}^T),\\ \end{array}
\ee
where $P_{\pm,1/2}$ are the projection operators:
\be\label{eq:pr}
P_{\pm,1/2}a_{i,j}:=\frac{1\pm {\rm sign}(j-i)}{2}a_{i,j}, \quad i, j=1,\dots,n,
\ee
defines an isomorphism between the Lie algebroid $(\mathfrak g,D_{\mathbb A})$ and the Lie algebroid
$\left(T^\ast \mathcal A,D_{\mathbb A}P_{\mathbb A}  \right)$.
\end{lm}

The Poisson bi-vector $\Pi$ on $\mathcal A$ is then obtained by the anchor map on the Lie algebroid
$\left(T^\ast \mathcal A,D_{\mathbb A}P_{\mathbb A}  \right)$ (see Proposition 10.1.4 in \cite{kirill}) as:
\be
\label{eq:biv}
\begin{array}{lccl}
\Pi:&T^\ast_{\mathbb A}\mathcal A\times T^\ast_{\mathbb A}\mathcal A
&\mapsto& \mathcal C^\infty(\mathcal A)\\
&(\omega_1,\omega_2)&\to&\Tr\left(\omega_1 D_{\mathbb A}P_{\mathbb A}  (\omega_2)
\right)
\end{array}
\ee
In coordinates one can compute the Poisson bracket by
\be
\label{Poisson-bracket}
\{a_{i,k},a_{j,l}\}:=\frac{\partial}{\partial{\rm d}a_{i,k}}\wedge\frac{\partial}{\partial{\rm d}a_{j,l}}
\Tr\left( {\rm d}a_{i,k} D_{\mathbb A}P_{\mathbb A}  ({\rm d}a_{j,l})\right).
\ee
This gives rise to the Poisson bracket on $\mathcal A$ given by the Dubrovin--Ugaglia bracket \cite{dubrovin,uga}:
\begin{eqnarray}\label{eq:du}
&&
\left\{a_{ik},a_{jl}\right\}=0,\quad\hbox{for}\,i<k<j<l,\hbox{ and }  i<j<l<k,\nn
\\&&
\left\{a_{ik},a_{jl}\right\}=2 \left(a_{ij}a_{kl}-a_{il}a_{kj}\right),\quad\hbox{for}\, i<j<k<l,
\\&&
\left\{a_{ik},a_{kl}\right\}=a_{ik}a_{kl}-2a_{il},\quad\hbox{for}\, i<k<l,\nn
\\&&
\left\{a_{ik},a_{jk}\right\}=-a_{ik}a_{jk}+2a_{ij},\quad\hbox{for} \, i< j<k,\nn
\\&&
\left\{a_{ik},a_{il}\right\}=-a_{ik}a_{il}+a_{kl},\quad\hbox{for} \, i<k<l.\nn
\end{eqnarray}

\begin{remark}\label{rm:key}
Note that the Poisson structure (\ref{eq:du}) is equivalent to the
one given by (\ref{Poisson}) by plugging in the restriction $\mathbb
A\in\mathcal A$ on the right-hand side.
\end{remark}

\subsection{The groupoid in the general case}

The key point in our construction is based on Remark \ref{rm:key}:
for any $n,m$ our algebra (\ref{Poisson}) is given by the same
Poisson bi-vector $\Pi$ as in the case $m=1$. So, due to equation
(\ref{eq:biv}), we must retain the same Lie algebroid structure on
$T^\ast\mathcal A_{n,m}$, in other words we keep the same anchor map
$D_{\mathbb A} P_{\mathbb A}$. Let us be more precise.

The tangent bundle $T{\mathcal A}_{n,m}$ is now identified with the
set of block--upper--triangular matrices $\delta \mathbb A$ such that
the diagonal blocks satisfy:
\be
\label{eq:diff-det}
\tr\mathbb A_{I,I}^{-1}\delta \mathbb A_{I,I}=0,
\ee
Analogously the cotangent bundle $T^\ast{\mathcal A}_{n,m}$ is now identified with the set of
block--lower--triangular matrices $\omega$ such that the diagonal blocks satisfy:
\be
\tr\mathbb A_{I,I}^{-1}\omega_{I,I}=0,
\ee
The map $P_{\mathbb A}$ is defined as above:
$$
\begin{array}{lccl}
P_{\mathbb A}:&T^\ast_{\mathbb A}\mathcal A_{n,m}&\to& Mat_{N}(\mathbb C)\\
& w &\mapsto &P_{-,1/2}(w{\mathbb A})-P_{+,1/2}(w^T{\mathbb A}^T),\\ \end{array}
$$
and has a non-trivial kernel now. We now define the Lie algebroid as image of this map:
\be\label{eq:algebroid-def}
\mathfrak g_{\mathbb A}:= {\rm Im}\left(P_{\mathbb A}\right),\qquad
\mathfrak g:= \cup_{\mathbb A\in\mathcal A_{n,m}}\mathfrak g_{\mathbb A}.
\ee

\begin{lm}
The bilinear form (\ref{eq:biv}) considered as a bilinear form on a vector space ${\mathbb C}^{(nm)^2}$
in which we substitute arbitrary (commuting) vectors $\omega_1,\omega_2 \in {\mathbb C}^{(nm)^2}$
is skew-symmetric.
\end{lm}

\proof
We first write the explicit expression for the bilinear form (\ref{eq:biv}):
\bea
\Tr\bigl[\omega_1 D_A\bigl(P_A(\omega_2)\bigr)\bigr]&=&
\Tr\Bigl[\omega_1 {\mathbb A}P_{-,1/2}(\omega_2 {\mathbb A}) - \omega_1 {\mathbb A} P_{+,1/2}(\omega_2^T{\mathbb  A}^T)\Bigr.\nonumber\\
&{}& \Bigl.+{\mathbb A}\omega_1 P_{+,1/2}({\mathbb A}^T\omega_2^T) - {\mathbb A} \omega_1 P_{-,1/2}({\mathbb A}\omega_2)\Bigr].
\label{form-omega12}
\eea
The assertion of the lemma follows if by substituting
$$
\omega_1=\omega_2=\omega
$$
in (\ref{form-omega12}) we obtain zero for any $\omega\in{\mathbb C}^{(nm)^2}$.
Using that
\be
\label{eq:decompose}
\begin{array}{l} \omega {\mathbb A}=P_{+,1/2}(\omega {\mathbb A})+P_{-,1/2}(\omega {\mathbb A}),\cr
   {\mathbb A} \omega =P_{+,1/2}({\mathbb A}\omega)+P_{-,1/2}({\mathbb A}\omega),\end{array}
\ee
for the sum of the second and third terms in the r.h.s. of (\ref{form-omega12})
we obtain under the trace sign the expression
\bea
&{}&-P_{+,1/2}(\omega {\mathbb A})P_{+,1/2}(\omega^T{\mathbb A}^T)-P_{-,1/2}(\omega {\mathbb A})P_{+,1/2}(\omega^T {\mathbb A}^T)\nonumber\\
&{}&+P_{+,1/2}({\mathbb A}\omega)P_{+,1/2}({\mathbb A}^T\omega^T)+P_{-,1/2}({\mathbb A}\omega)P_{+,1/2}({\mathbb A}^T\omega^T).\nonumber
\eea
Here the second term is the transposed fourth term (with the opposite sign), so the sum of these two terms
vanishes under the trace sign. In the first and third terms, only the products of diagonal projections, $P_d$, contribute
to the trace, and we obtain under the trace sign the expression
\bea
&{}&\quad-\frac14 P_d(\omega {\mathbb A})P_d(\omega^T {\mathbb A}^T)+\frac14 P_d({\mathbb A} \omega )P_d({\mathbb A}^T\omega^T)\nonumber\\
&{}&=-\frac14 P_d(\omega {\mathbb A})P_d(\omega^T {\mathbb A}^T)+\frac14 P_d({\mathbb A}^T\omega^T)P_d({\mathbb A} \omega )\nonumber\\
&{}&=-\frac14 P_d(\omega {\mathbb A})P_d(\omega^T {\mathbb A}^T)+\frac14 P_d(\omega {\mathbb A})P_d(\omega^T {\mathbb A}^T )=0,\nonumber
\eea
where we have used that, for any two matrices $X$ and $Y$, $P_d(X)P_d(Y)=P_d(Y)P_d(X)$ and $P_d(X^T)=P_d(X)$.

For the sum of the first and fourth terms in the r.h.s. of (\ref{form-omega12}) we obtain (using the cyclicity property of the trace)
\bea
&{}&\quad\Tr\bigl[\omega {\mathbb A} P_{-,1/2}(\omega {\mathbb A})-{\mathbb A}\omega P_{-,1/2}({\mathbb A}\omega)\bigr]\nonumber\\
&{}&=\Tr\bigl[-\omega {\mathbb A} P_{+,1/2}(\omega {\mathbb A})+{\mathbb A}\omega P_{+,1/2}({\mathbb A}\omega)\bigr]\label{eq:14}\\
&{}&=\frac12 \Tr\bigl[\omega {\mathbb A} [P_{-,1/2} -P_{+,1/2}](\omega {\mathbb A})
-{\mathbb A}\omega [P_{-,1/2} -P_{+,1/2}]({\mathbb A}\omega)\bigr].\nonumber
\eea
But, for any matrix $X$, we have that
\bea
&{}&\quad\Tr\bigl[X[P_- -P_+](X)\bigr]=\Tr\bigl[[P_- +P_+ +P_d](X)[P_- -P_+](X)\bigr]\nonumber\\
&{}&=\Tr\bigl[P_+(X)P_-(X) -P_-(X)P_+(X)\bigr]=0,
\eea
and each term in the last line of (\ref{eq:14}) therefore vanishes, which completes the proof.\endproof

\begin{theorem}
The triple $\left(\mathfrak g, D_{\mathbb A}, [\cdot,\cdot] \right)$
where the anchor map $D_{\mathbb A}$ is given by (\ref{gA}) and  the
Lie bracket $[\cdot,\cdot]$ on the space of sections
$\Gamma(\mathfrak g)$ is given by (\ref{eq:lieb}) is a Lie
algebroid.
\end{theorem}

\proof
The fact that the anchor map $D_{\mathbb A}$ satisfies the Leibnitz
rule with respect to the Lie bracket (\ref{eq:lieb}) is already
proved in \cite{Bon1} in the case $m=1$ and that proof extends
trivially to the case of arbitrary $m$. We only need to prove that
$\Gamma(\mathfrak g)$ is closed, i.e. that for any two sections
$v_1,v_2$, there exists $\omega\in T^\ast_{\mathbb A}\mathcal
A_{n,m}$ such that $[v_1,v_2](\mathbb A)= P_{\mathbb A}(\omega)$.
Take $\omega_i\in T^\ast_{\mathbb A}\mathcal A_{n,m}$ such that
$v_i(\mathbb A)=P_{\mathbb A}(\omega_i)$, $i=1,2$. The direct calculation
then yields $\omega$:
\bea
\omega&=&\varpi_{-,1}\Bigl[ P_{+,1/2}(\omega_2^T\mathbb A^T ) \omega_1-P_{+,1/2}(\omega_1^T\mathbb A^T ) \omega_2+
P_{+,1/2}(\omega_2 \mathbb A ) \omega_1  -P_{+,1/2}(\omega_1 \mathbb A ) \omega_2+\Bigr.\nn\\
&&
\Bigl.+\omega_2 P_{+,1/2}( \mathbb A^T\omega_1^T )-\omega_1 P_{+,1/2}( \mathbb A^T\omega_2^T )+
\omega_1 P_{-,1/2}( \mathbb A\omega_2 )-\omega_2 P_{-,1/2}( \mathbb A\omega_1 )\Bigr],
\nn
\eea
where we let $\varpi_{-,1}$ denote the (natural) projection on the set of (non-strictly) block-lower-triangular
matrices, $\varpi_{-,1}(GL_{nm})={\mathcal A}^T_{n,m}$.

This concludes the proof.
\endproof

We now integrate the Lie algebroid $\left(\mathfrak g, D_{\mathbb
A}, [\cdot,\cdot] \right)$ to obtain the Lie groupoid in which we
will have to pick the generators of the braid group.

\begin{theorem}
The Lie groupoid $\mathcal M_{n,m}$ which integrates the Lie algebroid defined by (\ref{eq:algebroid-def}) is given by
\be\label{groupoid1}
\mathcal M_{n,m}:= U_{\mathbb\mathcal A_{n,m}}\mathcal M^{(n,m)}_\mathbb A,
\ee
where
\be\begin{array}{l}
\mathcal M^{(n,m)}_\mathbb A:= \left\{B\in GL(\mathbb C^{n})\, | \,
 B{\mathbb A}B^T\in \mathcal A_{n,m}\hbox{ and the sets of central elements coincide:}\right.\\
 \left.
\qquad\qquad \{ b^{(I)}_d( B{\mathbb A}B^T)\}=\{b^{(I)}_d(\mathbb A)\}, d=0,\dots,[\frac{m}{2}], I=1,\dots,n
\right\},\label{groupoid2}
\end{array}\ee
where
$$
b_d^{(I)}\equiv\det M_{m-d}^{(I)}/\det M_d^{(I)},
$$
are the rational central elements of our algebras (\ref{Poisson}) and (\ref{eq:Yangian}).
\end{theorem}

\proof
Assume there exists a Lie groupoid  integrating the Lie algebroid defined by
(\ref{eq:algebroid-def}). Let us prove that it must then preserve all
central elements. Let $f$ be any central element for the algebra (\ref{Poisson}). Its variation along any element of the groupoid is
$$
\delta f= \sum_{ij}\frac{\partial f}{\partial a_{i,j}}\delta a_{i,j}=
\sum_{ij}\frac{\partial f}{\partial a_{i,j}}\left(D_{\mathbb A}P_{\mathbb A}(\omega)
\right)_{i,j}$$
for $\omega=\delta a_{j,i}$ (thanks to the Killing form). Using the definition (\ref{eq:biv}) of the Poisson bi-vector $\Pi$, we have that
$$
\delta f=\Pi({\rm d}f,\delta a_{j,i})=\{f,a_{i,j}\},
$$
where the right hand side is zero for a central element.
This shows that the Lie groupoid  integrating the Lie algebroid defined by (\ref{eq:algebroid-def})
must preserve all central elements and therefore it is defined by
(\ref{groupoid1}), (\ref{groupoid2}).
\endproof

We conclude this Section observing that  the identity morphism, the
inverse and the partial multiplication for the groupoid $\mathcal
M_{n,m}$ are still given by  (\ref{eq:emap}), (\ref{eq:imap}) and
(\ref{eq:mmap}) respectively.

\begin{remark}
We do not tackle the question whether or not there exists a smooth groupoid  on
$\mathcal A_{n,m}$ such that its Lie algebroid structure is given by $(T^\ast\mathcal A_{n,m},D_{\mathbb A}P_{\mathbb A})$.
The interested reader is invited to look at the beautiful work by Crainic and Fernandes \cite{CF1}.
\end{remark}

\section{Braid-group transformations}\label{s:braid}

The braid-group transformations $\beta_{I,I+1}$, $I=1,\dots,n-1$,
are transformations from the groupoid (\ref{groupoid1}), (\ref{groupoid2}) preserving the form of the matrix ${\mathbb A}$,
so by construction they must preserve the Poisson structure (\ref{Poisson}). They act of ${\mathbb A}$ as follows:
\be
\label{beta}
\beta_{I,I+1}[{\mathbb A}]=B_{I,I+1}{\mathbb A}B_{I,I+1}^T\equiv \wtd{\mathbb A},
\ee
where the matrix $B_{I,I+1}$ has the block form
\be
\label{BII+1}
B_{I,I+1}=\begin{array}{c}\vdots\cr I\cr I+1\cr \vdots \end{array}\left[\begin{array}{cccccccc}{\mathbb E}&&&&&&&\cr
&\ddots&&&&&&\cr&&{\mathbb E}&&&&&\cr&&&{\mathbb A}_{I,I+1}^T{\mathbb A}_{I,I}^{-T}&-{\mathbb E}&&&\cr
&&&{\mathbb A}_{I,I}{\mathbb A}_{I,I}^{-T}&{\mathbb O}&&&\cr
&&&&&{\mathbb E}&&\cr&&&&&&\ddots&\cr&&&&&&&{\mathbb E}
\end{array}\right],
\ee
as above, ${\mathbb E}$ and ${\mathbb O}$ are the respective $m\times m$ unit and zero matrices.

It is straightforward to verify that the transformation (\ref{beta}) preserves the form of the matrix ${\mathbb A}$
with
\be
\label{braid-for-AIJ}\begin{array}{l}
\wtd{{\mathbb A}_{I,I}}={\mathbb A}_{I+1,I+1},\quad \wtd{{\mathbb A}_{I+1,I+1}}={\mathbb A}_{I,I},\quad
\wtd{{\mathbb A}_{I,I+1}}={\mathbb A}_{I,I+1}^T\cr
J<I: \ \wtd{{\mathbb A}_{J,I}}={\mathbb A}_{J,I}{\mathbb A}_{I,I}^{-1}{\mathbb A}_{I,I+1}-{\mathbb A}_{J,I+1}, \quad
\wtd{{\mathbb A}_{J,I+1}}={\mathbb A}_{J,I}{\mathbb A}_{I,I}^{-1}{\mathbb A}_{I,I}^T,\cr
J>I+1:\ \wtd{{\mathbb A}_{I,J}}={\mathbb A}_{I,I+1}^T{\mathbb A}_{I,I}^{-T}{\mathbb A}_{I,J}-{\mathbb A}_{I+1,J}, \quad
\wtd{{\mathbb A}_{I+1,J}}={\mathbb A}_{I,I}{\mathbb A}_{I,I}^{-T}{\mathbb A}_{I,J}
\end{array}
\ee
and with all other blocks retaining their form.

We have two theorems concerning the transformations (\ref{beta}), (\ref{BII+1}).

\begin{theorem}\label{automorphism}
The transformations (\ref{beta}), (\ref{BII+1}) are automorphisms of the Poisson structure (\ref{Poisson}) restricted to the
block-upper-triangular matrices ${\mathbb A}$ from Definition~\ref{def-block}.
\end{theorem}

The statement follows from that the transformation (\ref{beta}), (\ref{BII+1}) is a transformation from the
groupoid of block-upper-triangular matrices.

\begin{theorem}
The transformations (\ref{beta}), (\ref{BII+1}) satisfy the braid-group relation,
\be
\label{braid-relation}
\beta_{I,I+1}\beta_{I+1,I+2}\beta_{I,I+1}[{\mathbb A}]
=\beta_{I+1,I+2}\beta_{I,I+1}\beta_{I+1,I+2}[{\mathbb A}],\ I=1,\dots,n-2.
\ee
\end{theorem}

We prove this theorem and the following proposition by the direct calculation.

\begin{prop}\label{prop-total-braid}
We have that $\bigl(\beta_{n-1,n}\cdots\beta_{2,3}\beta_{1,2}\bigr)^n[{\mathbb A}]=\wtd{\mathbb A}$, where
${\wtd {{\mathbb A}_{I,J}}}=\left({\mathbb A}_{I,I}{\mathbb A}_{I,I}^{-T}\right)^{n-2}{\mathbb A}_{I,J}
\left({\mathbb A}_{J,J}^{-1}{\mathbb A}_{J,J}^T\right)^{n-2}$ and,
in particular, ${\wtd {{\mathbb A}_{I,I}}}={\mathbb A}_{I,I}\ \forall n$.
\end{prop}

\subsection{Extension of the braid group action to the twisted Yangian case}

As in the case of the standard twisted Yangian algebra (see~\cite{ChM}), we have the extension of the braid-group action
in the case where the matrix ${\mathbb A}$ has the original block-upper-triangular form with all blocks having the same size $m\times m$.

\begin{prop}\label{prop-braid-Yangian}
The extended braid group transformations for the algebra (\ref{eq:Yangian}) in the case where the matrix ${\mathbb A}$
has the block-upper-triangular form described in Definition~\ref{def-block}
admits the following matrix representation in terms of the matrix ${\mathcal G}(\lambda)$ (\ref{G-def}):
\be\label{RA-infty}
\beta_{I,I+1}[{\mathcal G}(\lambda)]=B_{I,I+1}{\mathcal G}(\lambda)B^T_{I,I+1},\quad I=1,\dots,n-1
\ee
where the matrices $B_{I,I+1}$ have the form (\ref{BII+1}).

The action of $\beta_{n,1}$ is
\be
\beta_{n,1}[{\mathcal G}(\lambda)]={ B}_{n,1}(\lambda){\mathcal G}(\lambda)
\bigr({ B}_{n,1}(\lambda^{-1})\bigl)^T,
\label{R-n1-A-infty}
\ee
where the matrix ${ B}_{n,1}(\lambda)$ has the block form
\be\label{eq:bn1}
B_{n,1}(\lambda)=\left(\begin{array}{ccccc}{\mathbb O}&&&& \lambda {\mathbb A}_{n,n}{\mathbb A}_{n,n}^{-T}\cr
&{\mathbb E}&&&\cr&&\ddots&&\cr&&&{\mathbb E}&\cr -\lambda^{-1}{\mathbb E}&&&&\bigl[ {\mathbb G}_{n,1}^{(1)}\bigr]^T{\mathbb A}_{n,n}^{-T}
\end{array}\right)
\ee
in which ${\mathbb G}_{n,1}^{(1)}$ is the $m\times m$ block in the lower left corner of the $mn\times mn$ matrix ${G}^{(1)}$.
\end{prop}

\begin{theorem}\label{braid-Yangian}
The transformations (\ref{RA-infty}), (\ref{R-n1-A-infty}) satisfy the braid-group relation,
\be
\label{braid-relation-Yangian}
\beta_{I,I+1}\beta_{I+1,I+2}\beta_{I,I+1}[{\mathcal G}(\lambda)]
=\beta_{I+1,I+2}\beta_{I,I+1}\beta_{I+1,I+2}[{\mathcal G}(\lambda)],\ I=1,\dots,n\ \mod n.
\ee
\end{theorem}

\section{Quantisation}

The affine algebra (\ref{eq:Yangian}) is the semiclassical limit of the quantum algebra generated by the matrix elements
$G_{i,j}^{(p)}$, $i,j=1,\dots,N$, $p\in\mathbb Z_{\geq 0}$ subject to the
defining relations:
\be\label{eq:molev}
R(\lambda,\mu)\one{\mathcal G}(\lambda)R(\lambda^{-1},\mu)^{T_1}\two{\mathcal G}(\mu)=
\two{\mathcal G}(\mu)R(\lambda^{-1},\mu)^{T_1}\one{\mathcal G}(\lambda)R(\lambda,\mu)
\ee
where the apex $T_1$ indicates the transposition in space one and the R-matrix is given by
\begin{eqnarray}\label{eq:qR}
R(\lambda,\mu)&=&(\lambda-\mu)\sum_{i\neq j} E_{ii}\otim E_{jj}+
(q^{-1}\lambda-q\mu)\sum_{i}E_{ii}\otim E_{ii} +\nn \\
&+&(q^{-1}-q)\lambda\sum_{i<j}E_{ij}\otim E_{ji}+(q^{-1}-q)\mu\sum_{i>j} E_{ij}\otim E_{ji}
\end{eqnarray}
it is a solution of the Yang--Baxter equation.

In coordinates the quantum algebra relations are pretty cumbersome,
let us present here the formula for the level $1$ reduction of the quantum algebra,
or in other words, for the quantum analogue of our algebra (\ref{Poisson}):
\bea\label{eq:molev1}
q^{\delta_{s,j}+\delta_{i,j}}a_{i,s} a_{j,t} -q^{\delta_{s,t}+\delta_{i,t}} a_{j,t}  a_{i,s}& =&
\bigr(q-q^{-1}\bigl) q^{\delta_{s,i}}(\delta_{t>s}-\delta_{i>j})a_{j,s} a_{i,t} +\nn\\
&+&\bigr(q-q^{-1}\bigl) \left(q^{\delta_{s,t}}\delta_{t>i} a_{j,i} a_{t,s}-
 q^{\delta_{i,j}}\delta_{s>j} a_{i,j} a_{s,t}\right)+\\
 &+&\bigr(q-q^{-1}\bigl)^2\delta_{s>i}\left(\delta_{t>s}-\delta_{i>j}\right)a_{j,i} a_{s,t} ,
\nn\eea
where $\delta_{i>j}=1$ for $i>j$ and $0$ otherwise. For $m=2$ this quantum algebra coincides with the
twisted quantised enveloping algebra $U^{tw}(\mathfrak{sp}_{2n})$ \cite{noumi,MRS}.

The affine quantum algebra (\ref{eq:molev}) coincides in the case of $m=1$ ($m=2$)
with the twisted $q$-Yangian $Y'_q(\mathfrak{o}_n)$ ($Y'_q(\mathfrak{sp}_{2n})$) for the
orthogonal (symplectic) Lie algebra introduced in \cite{MRS}.
For $m>2$ this algebra has never been studied before to the best of our knowledge.

In the semiclassical limit the affine quantum algebra (\ref{eq:molev})  gives rise to our affine algebra (\ref{eq:Yangian}).
We already calculated this semiclassical limit  in the case $m=1$ in \cite{ChM}, for general $m$ the computation is exactly the same.

In this paper, we construct the quantum braid-group action only in the case of the $\mathfrak{sp}_{2n}$ case, but the main
features of the technique must remain unchanged for both the general $m\times m$-b.u.t. case and for the affine algebras.

\subsection{Quantum braid group action for the $\mathfrak{sp}_{2n}$ case}

In order to quantise the braid group action, we need to find the quantum analogues of the inverse,
$\mathbb A^{-1}_{I,I}$, the transposed, $\mathbb A^{T}_{I,I}$, and the
inverse-transposed, $\mathbb A^{-T}_{I,I}$, matrices for the diagonal blocks and also the transposed
$\mathbb A^{T}_{I,I+1}$ for the off-diagonal blocks.

We first find the laws of  quantum complex conjugation for all the entries $a_{i,j}$
(these formulas are valid for all $n$ and $m$). The main point is that the quantum complex conjugation
needs to be an automorphism of the algebra (\ref{eq:molev}).

From formulas (\ref{eq:molev1}), assuming all the diagonal entries $a_{i,i}$, $i=1,\dots,mn$, to be self-adjoint operators,
we obtain the laws of conjugation:
\be
a^*_{i,i}=a_{i,i},\quad \left\{ \begin{array}{l} a^*_{j,i}=q a_{j,i}\cr a^*_{i,j}=q a_{i,j}+(1-q^2)a_{j,i}\end{array}\ \hbox{for}\ i<j.\right.
\label{conjugation}
\ee
The last formula implies that in the case where we restrict to the b.u.t. case,  the lower-triangular ($i<j$) matrix entries
$a_{j,i}$ not belonging to the diagonal blocks
vanish, and we obtain merely that $a^*_{i,j}=q a_{i,j}$ for all the entries of the matrices $\mathbb A_{I,J}$ with $I<J$.

We then have the following prescription:

(1) All the transposition operations are replaced by the Hermitian conjugations.

Note that the Hermitian conjugation of the diagonal blocks is different
from the Hermitian conjugation of the off-diagonal ones. As an example,
we present the result of the Hermitian conjugation for the block
$\mathbb A_{1,1}$,
\be
\left(\begin{array}{cc}
        a_{11} & a_{12} \\
        a_{21} & a_{22}
      \end{array}
\right)^{\dagger}=
\left(\begin{array}{cc}
        a_{11} &  qa_{21}  \\
        qa_{12}+(1-q^2)a_{21} & a_{22}
      \end{array}
\right),
\label{A11-conjugate}
\ee
and
\be
\label{transpose-IJ}
\mathbb A_{I,J}^\dagger =q [\mathbb A_{I,J}]^T\ \hbox{for}\ J>I,
\ee
where the transposition is understood here and hereafter in the standard matrix sense
(note however that the transpose of the product of two matrices in the quantum case
is not given by the reverse order product of their transposed due to the noncommutativity of
their entries, say, $(\mathbb A_{I,J}\mathbb A_{K,L})^T\neq
\mathbb A_{K,L}^T\mathbb A_{I,J}^T$).

(2) The inverse $\mathbb A_{I,I}^{-1}$ is to be found from the operatorial identity
$\mathbb A_{I,I}\mathbb A_{I,I}^{-1}={\mathbb E}$, in which
the order in which we multiply operators follows from that of the matrix multiplication. In the $m=2$ case,
the result is (we present it for $\mathbb A_{1,1}$, the generalisation to other $2\times 2$-blocks
$\mathbb A_{I,I}$ is obvious)
\be
\left(\begin{array}{cc}
        a_{11} & a_{12} \\
        a_{21} & a_{22}
      \end{array}
\right)^{-1}=\frac1{a_{11}a_{22}-q^2 a_{12}a_{21}}
\left(\begin{array}{cc}
        a_{22} & -a_{12}+(q-q^{-1})a_{21} \\
        -q^2a_{21} & a_{11}
      \end{array}
\right)
\label{A11-inverse}
\ee
Here the combination $a_{11}a_{22}-q^2 a_{12}a_{21}$ is the quantum determinant of the block
$\mathbb A_{1,1}$; all these
determinants are self-adjoint central elements of the quantum algebra $\mathfrak{sp}_{2n}$.

(3) eventually, the inverse-transposed in the quantum case becomes
$[\mathbb A_{I,I}^{-1}]^\dagger=\bigl[\mathbb A_{I,I}^\dagger\bigr]^{-1}$
where we use the expressions (\ref{A11-inverse}) and (\ref{A11-conjugate}); the result is merely (we always assume
that $q:=e^{-i\pi\hbar}$ and, therefore, $q^*=q^{-1}$)
\be
\left(\begin{array}{cc}
        a_{11} & a_{12} \\
        a_{21} & a_{22}
      \end{array}
\right)^{-\dagger}=\frac1{a_{11}a_{22}-q^2 a_{12}a_{21}}
\left(\begin{array}{cc}
        a_{22} & -q^{-1}a_{21}  \\
        -qa_{12} & a_{11}
      \end{array}
\right)
\label{A11-inverse-conjugate}
\ee

\begin{theorem}\label{th:quantum-braid}
For the b.u.t matrix ${\mathbb A}^\hbar$ from Definition~\ref{def-block} whose entries are operators subject to the
conjugation law (\ref{conjugation}) the action of the quantum braid group has the adjoint matrix form
\be
\label{beta-q}
\beta_{I,I+1}^\hbar[{\mathbb A}^\hbar]=B^\hbar_{I,I+1}{\mathbb A}^\hbar \Bigl[B_{I,I+1}^\hbar\Bigr]^\dagger\equiv \wtd{{\mathbb A}^\hbar},
\ee
where the matrix $B^\hbar_{I,I+1}$ has the block form (here and hereafter we assume that all matrix entries are operators and
that these operators are multiplied in the order prescribed by the matrix multiplication)
\be
\label{BII+1-q}
B^\hbar_{I,I+1}=\begin{array}{c}\vdots\cr I\cr I+1\cr \vdots \end{array}\left(\begin{array}{cccccccc}{\mathbb E}&&&&&&&\cr
&\ddots&&&&&&\cr&&{\mathbb E}&&&&&\cr&&&q^{-a}{\mathbb A}_{I,I+1}^\dagger{\mathbb A}_{I,I}^{-\dagger}&-q^{-a}{\mathbb E}&&&\cr
&&&q^{-b} {\mathbb A}^{}_{I,I}{\mathbb A}_{I,I}^{-\dagger}&{\mathbb O}&&&\cr
&&&&&{\mathbb E}&&\cr&&&&&&\ddots&\cr&&&&&&&{\mathbb E}
\end{array}\right),
\ee
where $a=b+1$ for all $m$ and $b=0$ for $m=1$ and $b=1$ for $m=2$.
\end{theorem}

Before proving the theorem, we formulate the following lemma, which can be verified by a simple calculation.

\begin{lm}\label{lemma-q-braid}
The two quantum braid-group relations
\be
\label{q-braid}
\beta_{I,I+1}^\hbar\beta_{I-1,I}^\hbar\beta_{I,I+1}^\hbar[{\mathbb A}^\hbar]
=\beta_{I-1,I}^\hbar\beta_{I,I+1}^\hbar\beta_{I-1,I}^\hbar[{\mathbb A}^\hbar],\quad I=2,\dots,n-1,
\ee
and
\be
\label{q-total-braid}
\bigl(\beta_{n-1,n}^\hbar\beta_{n-2,n-1}^\hbar\cdots\beta_{2,3}^\hbar\beta_{1,2}^\hbar\bigr)^n[{\mathbb A}_{I,J}]
=\bigl({\mathbb A}_{I,I}^{}{\mathbb A}_{I,I}^{-\dagger}\bigr)^{n-2}{\mathbb A}_{I,J}
\bigl({\mathbb A}_{J,J}^{-1}{\mathbb A}_{J,J}^{\dagger}\bigr)^{n-2},
\ee
which are the quantum analogues of the respective relation (\ref{braid-relation}) and Proposition \ref{prop-total-braid},
are satisfied for any choice of the parameters $a$ and $b$ in the formula (\ref{BII+1-q}).
\end{lm}

\proof{\it of Theorem~\ref{th:quantum-braid}}.\quad We need to prove
that the quantum algebra (\ref{eq:molev1}) is preserved under the
quantum action (\ref{beta-q}). It is enough to restrict to the case
of $n=3$ and concentrate on the action of $\beta_{1,2}^\hbar$:
$$
\beta_{1,2}^\hbar(\mathbb A^\hbar)=\wtd{\mathbb A}^\hbar=\left(\begin{array}{ccc}
\mathbb A_{2,2}&q^{-a+b} \mathbb A_{1,2}^\dagger &
q^{-a}(\mathbb A_{1,2}^\dagger\mathbb A_{1,1}^{-\dagger}\mathbb A_{1,3}-\mathbb A_{2,3})\\
\mathbb O&\mathbb A_{1,1}&q^{-b}(\mathbb A_{1,1}\mathbb A_{1,1}^{-\dagger}\mathbb A_{1,3})\\
\mathbb O&\mathbb O&\mathbb A_{3,3}\\
\end{array}\right).
$$
In order for the $(1,2)$ matrix entry to have the same conjugation law (\ref{transpose-IJ}) as the original operator ${\mathbb A}_{1,2}$
we need that $-a+b=-1$, or $a=b+1$ for all $m$. However, when considering the $(2,3)$ matrix entry, we observe the explicit difference
between cases where $m=1$ and $m=2$. If $m=1$, then ${\mathbb A}_{1,1}=1$, and we just have $q^{-b}{\mathbb A}_{1,3}$, so, in order to preserve the
conjugation law we must merely set $b=0$. But in the case of $2\times 2$-matrices the situation becomes different: we need a rather
nontrivial calculation, which involves repeated applications of the commutation relations (\ref{eq:molev1}), to demonstrate that
$$
\bigl({\mathbb A}^{}_{1,1}{\mathbb A}_{1,1}^{-\dagger}{\mathbb A}_{1,3}\bigr)^\dagger
=q^{-1}\bigl[{\mathbb A}_{1,1}^{}{\mathbb A}_{1,1}^{-\dagger}{\mathbb A}_{1,3}\bigr]^T\ \hbox{for}\ m=2.
$$
So, we need to set $b=1$ for $m=2$ in order to preserve the conjugation law.

The proof of preserving the algebra is then a straightforward brute force computation in which one needs to check relations (\ref{eq:molev1})
entry by entry.
\endproof

\begin{remark}\label{rem:autom}
The transformation $\beta_{1,2}$ given by the formula (\ref{beta-q}) in the case of $2\times 2$-block matrix ${\mathbb A}=
\left(\begin{array}{cc}
        {\mathbb A}_{1,1} & {\mathbb A}_{1,2}  \\
        {\mathbb O} & {\mathbb A}_{2,2}
      \end{array}
\right)$ results in that ${\mathbb A}_{1,1}\to {\mathbb A}_{2,2}$, ${\mathbb A}_{2,2}\to {\mathbb A}_{1,1}$, and
${\mathbb A}_{1,2}\to \bigl[{\mathbb A}_{1,2}\bigr]^T$. It is easy to see from the algebra (\ref{eq:molev}) that this
transformation is an automorphism of the corresponding quantum algebra. This automorphism was stated as Theorem~4.8 in~\cite{MR}.
\end{remark}

\begin{conjecture}
It is plausible that the braid-group action in the general case of $m\times m$-matrix blocks has the same
form (\ref{beta-q}) and (\ref{BII+1-q}) with $a=b+1$ and $b=m-1$ where we have to determine the inverse matrix ${\mathbb A}_{I,I}^{-1}$ from the
operatorial equality ${\mathbb A}_{I,I}{\mathbb A}_{I,I}^{-1}={\mathbb E}$ and use the conjugation rules (\ref{conjugation}).
\end{conjecture}

\begin{conjecture}
In the twisted Yangian case (\ref{eq:molev}) for $m=2$ we expect to extend the same quantum
conjugation relations to all levels and to quantise the element $\beta_{n,1}$ defined by (\ref{eq:bn1})  to
$$
\beta_{n,1}^\hbar[{\mathbb A}^\hbar]=B^\hbar_{n,1}(\lambda){\mathbb A}^\hbar \Bigl[B_{n,1}^\hbar(\lambda^{-1})\Bigr]^\dagger,$$
where the matrix $B^\hbar_{n,1}(\lambda)$ has the block form:
$$
B_{n,1}^\hbar(\lambda)=\left(\begin{array}{ccccc}{\mathbb O}&&&& \lambda q^{-b}{\mathbb A}^{}_{n,n}{\mathbb A}_{n,n}^{-\dagger}\cr
&{\mathbb E}&&&\cr&&\ddots&&\cr&&&{\mathbb E}&\cr
-\lambda^{-1}q^{-b-1}{\mathbb E}&&&&q^{-b-1}\bigl[ {\mathbb G}_{n,1}^{(1)}\bigr]^\dagger{\mathbb A}_{n,n}^{-\dagger}
\end{array}\right)
$$
with $b=m-1$.
\end{conjecture}

\subsection{$R$-matrix role in the classical case}

The classical Poisson bracket (\ref{Poisson}) was obtained by imposing
\be\label{eq:PBd}
\{a_{i,j},a_{k,l}\}=
\frac{\partial}{\partial{\rm d}a_{i,j}}\wedge\frac{\partial}{\partial{\rm d}a_{k,l}} S,
\ee
where $S$ is the Poisson bi-vector computed on the one forms ${\rm d}a_{i,j}$ and ${\rm d}a_{k,l}$:
$$
S=\Tr\left({\rm d}a_{i,j}D_{\mathbb A}P_{\mathbb A}({\rm d} a_{k,l})\right).
$$
In this subsection we want to show that  (\ref{eq:PBd}) can also be written in $R$--matrix notation as:
\be\label{eq:PR}
\{\one{\mathbb A},\two{\mathbb A}\}=-\left(\two{\mathbb A}\one{\mathbb A} r-
r\one{\mathbb A}\two{\mathbb A}+\two{\mathbb A}r^{T_1}\one{\mathbb A}-
\one{\mathbb A}r^{T_1}\two{\mathbb A}\right),
\ee
where $r$ is given by
$$
r=\sum_{i} E_{ii}\otimes E_{ii}+2 \sum_{i>j}E_{ij}\otimes E_{ji}.
$$
Of course, we could just do a brute force computation to simply prove that the two formulae coincide.
However,  we prefer to show a more general proof which relies on the
relation between the $P_{\mathbb A}$ operator and the classical $R$-matrix.

Given any matrix $X$, let us define the following four operators
$\one r_{\pm}$, $\two r_{\pm}$:
\bea\label{r-operators}
&&
\one r_{+}(X):=\frac{1}{2} \two\Tr\left(r^T \two X\right),\qquad
\two r_{+}(X):=\frac{1}{2} \one\Tr\left(r \one X\right),\\
&&
\one r_{-}(X):=\frac{1}{2} \two\Tr\left(r \two X\right),\qquad
\two r_{-}(X):=\frac{1}{2}  \one\Tr\left(r^T \one X\right),\nn
\eea
then our matrix $g:=P_{\mathbb A}(w)$ is given by two different formulae according to which space we want it in:
\be\label{eq:par1}
\one g=\one r_{-}\left(\two w \two{\mathbb A}\right)-\one r_{+}\left(\two w^T \two{\mathbb A}^T\right),
\ee
\be\label{eq:par2}
\two g=\two r_{-}\left(\one w \one{\mathbb A}\right)-\two r_{+}\left(\one w^T \two{\mathbb A}^T\right),
\ee
where $r^T$ is the transposition in both spaces.

Now, it is clear that in $R$-matrix notation $S$ is given by
\bea
S&=&
\Tr\bigl[(d{\mathbb A})^T{\mathbb A}r_{-}[(d{\mathbb A})^T{\mathbb A}]-
(d{\mathbb A})^T{\mathbb A}r_{+}[d{\mathbb A}{\mathbb A}^T]\bigr.
\nonumber\cr
&{}&+\bigl.{\mathbb A}(d{\mathbb A})^T r_{+}[{\mathbb A}^Td{\mathbb A}]
-{\mathbb A}(d{\mathbb A})^T r_{-}[{\mathbb A}(d{\mathbb A})^T]\bigr]\nn
\eea
where we did not specify the spaces as any choice leads to the same result (up to sign). For example choosing space $2$ we get:
\bea
S&=& \frac{1}{2}\onetwo{\Tr}\bigl[
(d\two{\mathbb A})^T\two{{\mathbb A}}\one {\mathbb A} \, r^T\, (d\one{\mathbb A})^T
-(d\two {\mathbb A})^T\, r^T\,  \two{\mathbb A}\one {\mathbb A}(d\one{\mathbb A})^T-\nn\\
&{}&-(d\one {\mathbb A})^T\two {\mathbb A}\, r^{T_1}\,\one{\mathbb A}(d\two{\mathbb A})^T+
(d \one{\mathbb A})^T\one{\mathbb A}\,r^{T_1}\,\two{\mathbb A}(d\two {\mathbb A})^T \bigr],
\label{symplecticR2}
\eea
so that
\bea\label{eq:final1}
\frac{\partial}{\partial{\rm d}\one{\mathbb A}}\wedge\frac{\partial}{\partial{\rm d}\two{\mathbb A}} S& = &
-\two {\mathbb A}\one{{\mathbb A}} \, r^T+ r^T\,  \one{\mathbb A}\two {\mathbb A}+
\one {\mathbb A}\, r^{T_1}\,\two{\mathbb A}-\two{\mathbb A}\,r^{T_1}\,\one{\mathbb A}=\\
&=&
-\left(\two {\mathbb A}\one{{\mathbb A}} \, r- r\,  \one{\mathbb A}\two {\mathbb A}-
\one {\mathbb A}\, r^{T_1}\,\two{\mathbb A}+\two{\mathbb A}\,r^{T_1}\,\one{\mathbb A}\right),\nn
\eea
as we wanted.

This is a rather interesting fact as it allows to pursue the same sort of approach based on the
construction of the morphism  groupoid  for any Poisson algebra which admits $R$-matrix formulation.
Indeed, we could define the $P_{\mathbb A}$ operator by (\ref{eq:par1}) or (\ref{eq:par2}), and
construct the Lie algebroid for any given $R$. This would allow us  to obtain the corresponding groupoid
in which to find the generators of the braid group. This work will be carried out in subsequent publications.
In particular it would be interesting to start from the chase of the exchange $R$--matrix
in which all central elements are known \cite{FM}, incidentally they are combinations of our rational
central elements  $b_d^{(I)}$ and their top right analogues.

\end{document}